\documentclass{article}
\usepackage{spconf,amsmath,graphicx}

\usepackage{booktabs}
\usepackage{mdwmath}
\usepackage{mdwtab}
\usepackage{graphicx,enumerate}
\usepackage{caption}
\usepackage{subcaption}
\usepackage{graphicx}
\graphicspath{{figures/}}
\usepackage{mathtools,commath,amsmath,amssymb,amsfonts,bm}\interdisplaylinepenalty=2500
\usepackage{amsthm}
\usepackage{tabularx,adjustbox}
\usepackage{todonotes}
\usepackage{breqn}
\usepackage{color,cite}
\clearpage
\extrafloats{100}
\RequirePackage{tcolorbox}
\tcbuselibrary{breakable}
\tcbset{parbox=false, left=1ex, right=1ex}

\usepackage{algorithmic}
\usepackage{algorithm}
\usepackage{multirow}
\usepackage{xspace}
\usepackage{url}
\newcolumntype{C}[1]{>{\hsize=#1\hsize\centering\arraybackslash}X}

\def\vec#1{\mathchoice{\mbox{\boldmath$\displaystyle#1$}}
  {\mbox{\boldmath$\textstyle#1$}}
  {\mbox{\boldmath$\scriptstyle#1$}}
  {\mbox{\boldmath$\scriptscriptstyle#1$}}}

\newcommand{\mech}{\ensuremath{\mathcal{M}}\xspace}

\newtheorem{theorem}{Theorem}

\newtheorem{Definition}{Definition}

\usepackage{hyperref}

\newtheorem{Remark}{Remark}

\title{DP-SIGNSGD: WHEN EFFICIENCY MEETS Privacy and Robustness}
\name{Lingjuan Lyu
}
\address{Ant Group}

\begin{document}
%
\maketitle
\begin{sloppypar}
\begin{abstract}
Federated learning (FL) has emerged as a promising collaboration paradigm by enabling a multitude of parties to construct a joint model without exposing their private training data. Three main challenges in FL are efficiency, privacy, and robustness. The recently proposed {\scriptsize SIGN}SGD with majority vote shows a promising direction to deal with efficiency and Byzantine robustness. However, there is no guarantee that {\scriptsize SIGN}SGD is privacy-preserving. In this paper, we bridge this gap by presenting an improved method called {\scriptsize DP-SIGN}SGD, which can meet all the aforementioned properties. We further propose an error-feedback variant of {\scriptsize DP-SIGN}SGD to improve accuracy. Experimental results on benchmark image datasets demonstrate the effectiveness of our proposed methods. 
\end{abstract}
\section{Introduction}
\label{sec:intro}

Federated learning (FL) provides a promising learning paradigm by pushing model training to local parties~\cite{mcmahan2016communication}. In the centralized optimization, communication costs are relatively small, and computational costs dominate. In contrast, in federated optimization, communication costs dominate. Up to now, it is still difficult to deploy practical FL applications in real world. To alleviate communication burden on the network, one promising solution is gradient quantization, such as SignSGD~\cite{bernstein2018signsgd1,bernstein2018signsgd2}, QSGD~\cite{alistarh2017qsgd}, TernGrad~\cite{wen2017terngrad} and ATOMO \cite{wang2018atomo}.

Meanwhile, FL offers a privacy-aware paradigm of model training which does not require data sharing. 
Nevertheless, recent works have demonstrated that FL may not always provide sufficient privacy guarantees, as communicating model updates throughout the training process can nonetheless reveal sensitive information~\cite{melis2019exploiting} even incur deep leakage~\cite{zhu2019deep}. For instance, as shown by~\cite{aono2018privacy}, even a small portion of gradients may reveal information about local data. A more recent work showed that the training data can be completely revealed from gradients in a few iterations when parties share their gradients of small batch~\cite{zhu2019deep}. Such attacks pose significant threats to FL~\cite{lyu2020threats,lyu2020privacy}, as in FL, any party may violate the privacy of other parties in the system, even without involving the server. In addition to privacy, there may exist Byzantine attackers in  FL systems, who can run data or model poisoning attacks to compromise the integrity of the learning process. 

As far as we know, there is no existing works that can satisfy multiple goals in FL simultaneously: (1) fast algorithmic convergence; (2) good generalisation performance; (3) communication efficiency; (4) fault tolerance; (5) privacy preservation. {\scriptsize SIGN}SGD with majority vote~\cite{bernstein2018signsgd1,bernstein2018signsgd2} offers a promising solution to realize the first four goals, but lacks privacy guarantee. To fill in this gap, we are inspired to add carefully calibrated noise before each party takes the sign operation~\cite{chen2019distributed}. In summary, our contributions include:
\begin{enumerate}
    \item We propose an efficient, privacy-preserving, and Byzantine robust compressor, which extends {\scriptsize SIGN}SGD to {\scriptsize DP-SIGN}SGD. We further provide an error-feedback version called {\scriptsize EF-DP-SIGN}SGD to improve accuracy.
    \item We theoretically prove the privacy guarantee of our proposed algorithms, and empirically demonstrate that our {\scriptsize DP-SIGN}SGD and {\scriptsize EF-DP-SIGN}SGD can ensure Byzantine robustness, accuracy and communication efficiency simultaneously.
\end{enumerate}

The remaining of the paper is organized as follows. Section~\ref{sec:related} reviews the related work  and preliminary in this work. Section~\ref{sec:DP-SIGNSGD} presents the technical details of our proposed {\scriptsize DP-SIGN}SGD. Section~\ref{sec:Performance} provides the empirical results and analysis. Finally, Section~\ref{sec:Conclusion} concludes this paper. 

\section{Related Works and Preliminary}
\label{sec:related}
\textbf{Communication Efficiency}: Sharing high-dimensional gradients across iterative rounds in federated learning is very costly. To reduce communication cost, various gradient quantization methods have been proposed. Algorithms like QSGD~\cite{alistarh2017qsgd}, TernGrad~\cite{wen2017terngrad} and ATOMO \cite{wang2018atomo} use stochastic quantisation schemes to ensure that the compressed stochastic gradient remains an unbiased approximation to the true gradient. More heuristic algorithms like 1BITSGD~\cite{seide20141} focuses more on practical performance. {\scriptsize SIGN}SGD with majority vote takes a different approach by directly employing the sign of the stochastic gradient~\cite{bernstein2018signsgd1,bernstein2018signsgd2}. For the homogeneous data distribution scenario, \cite{bernstein2018signsgd1,bernstein2018signsgd2} show theoretical and empirical evidence that {\scriptsize SIGN}SGD can converge well despite the biased approximation nature. For the heterogeneous data distribution, \cite{chen2019distributed} shows that the convergence of {\scriptsize SIGN}SGD is not guaranteed and proposes to add carefully designed noise to ensure convergence. However, compared with the full-precision gradients, {\scriptsize SIGN}SGD always incurs extra quantization error. To remedy this issue, \cite{seide20141} tracks quantisation errors and feeds them back into the subsequent updates. \cite{karimireddy2019error} proposes {\scriptsize EF\text{-}SIGN}SGD, which applies error compensation to {\scriptsize SIGN}SGD; however, this work only considered the single party scenario. \cite{zheng2019communication} further extends it to the multi-party scenario and establishes the convergence. 

\textbf{Privacy and Byzantine Robustness:} 
Previous works have shown that sharing gradients can result in serious privacy leakage~\cite{aono2018privacy,melis2019exploiting,zhu2019deep}. In addition to privacy, FL robustness is a big challenge that may hinder the applicability of FL, as there may exist Byzantine attackers in addition to the normal parties. FL with secure aggregation~\cite{bonawitz2017practical} is especially susceptible to poisoning attacks as the individual updates cannot be inspected. The appearance of {\scriptsize SIGN}SGD provides a promising method for both privacy and Byzantine robustness. {\scriptsize SIGN}SGD can largely reduce privacy leakage via gradients quantization. For robustness, \cite{bernstein2018signsgd2} shows that {\scriptsize SIGN}SGD can tolerate up to half ``blind" Byzantine parties who determine how to manipulate their gradients before observing the gradients. However, there is no theoretic privacy guarantee of {\scriptsize SIGN}SGD. 

\textbf{Local Differential Privacy (LDP):} In terms of privacy guarantee, a formal definition of LDP is provided in Definition~\ref{def:ldp}. The privacy budget $\epsilon$ captures the privacy loss consumed by the output of the algorithm: $\epsilon=0$ ensures perfect privacy in which the output is independent of its input, while $\epsilon\rightarrow\infty$ gives no privacy guarantee.

\begin{Definition}
\label{def:ldp}
Let $\mech: \mathcal{D} \to \mathcal{O}$ be a randomised algorithm mapping a data entry in $\mathcal{D}$ to $\mathcal{O}$. The algorithm $\mech$ is $(\epsilon,\delta)$-local differentially private if for all data entries $\vec{x}, \vec{x'} \in \mathcal{D}$ and all outputs $o \in \mathcal{O}$, we have
\begin{equation*}
\Pr\{\mech(\vec{x})=o\} \leq \exp(\epsilon) \Pr\{\mech(\vec{x'})=o\} + \delta
\end{equation*}
If $\delta=0$, $\mech$ is said to be $\epsilon$-local differentially private.
\end{Definition}

For every pair of adjacent inputs $\vec{x}$ and $\vec{x'}$, LDP requires that the distribution of $\mech(\vec{x})$ and $\mech(\vec{x'})$ are ``close'' to each other where closeness are measured by the privacy parameters $\epsilon$ and $\delta$. 

To the best our knowledge, none of the previous works can ensure efficiency, privacy, Byzantine robustness, accuracy and convergence simultaneously.

\section{{\scriptsize DP-SIGN}SGD}
\label{sec:DP-SIGNSGD}
In FL, at each communication round $t$, each party $m \in [M]$ computes local gradient $\boldsymbol{g}_{m}^{(t)}$ and sends it to the server. The server performs aggregation and sends the aggregated gradient back to all parties~\cite{mcmahan2016communication}. Finally, parties update their local models using the aggregated gradient. 

Considering communication efficiency and privacy, we adopt {\scriptsize SIGN}SGD and combine Gaussian Mechanism~\cite{dwork2014algorithmic} with {SIGN}SGD. More specifically, we employ the state-of-the-art \emph{Analytic Gaussian Mechanism} for each party in the FL system, which relaxes the constraint of $\epsilon<1$ in the traditional Gaussian Mechanism~\cite{dwork2014algorithmic} as follows.

\begin{theorem}
\label{Theorem:AGM}
(Analytic Gaussian Mechanism~\cite{balle2018improving}). Let $f: \mathbb{X} \rightarrow \mathcal{R}^d$ be a function with global $L_2$ sensitivity $\Delta$. For any $\epsilon \geq 0$ and $\delta \in [0,1]$, the Analytic Gaussian Mechanism $M(x) = f(x) + Z$ with $Z \sim N (0,\sigma^2 I)$ is $(\epsilon, \delta)$-DP if and only if

\begin{equation}\label{eq:Analytic_Gaussian}
\Phi \left(\frac{\Delta}{2\sigma}-\frac{\epsilon\sigma}{\Delta}\right) - e^\epsilon \Phi \left(-\frac{\Delta}{2\sigma}-\frac{\epsilon\sigma}{\Delta}\right) \leq \delta.
\end{equation}
\end{theorem}

In order to obtain $(\epsilon, \delta)$-DP for a function $f$ with global $L_2$ sensitivity $\Delta$, it is enough to add Gaussian noise with variance $\sigma^2$ satisfying Equation~\ref{eq:Analytic_Gaussian}. 

\subsection{The Differentially Private Compressor $dpsign$}
In this subsection, we present the differentially private version of {SIGN}SGD. Instead of sharing actual local gradient $\boldsymbol{g}_{m}^{(t)}$, each party $m$ quantizes the gradient with a differentially private 1-bit compressor $dpsign(\cdot)$ and sends $dpsign(\boldsymbol{g}_{m}^{(t)})$ to the server by following Equation~\ref{dpsignsgd}. The differentially private compressor $dpsign$ is formally defined as follows. The probability of each coordinate of the gradients mapping to $\{-1,1\}$ is sophistically designed to satisfy the LDP guarantee.

The main procedures for {\scriptsize DP-SIGN}SGD and {\scriptsize EF-DP-SIGN}SGD are given in Algorithm~\ref{algo:DP-SIGNSGD} and Algorithm~\ref{algo:Error-feedback-DP-SIGNSGD} respectively.
In particular, we introduce error decay rate $\lambda$ to incorporate error compensation.

\begin{Definition}
For any given gradient $\boldsymbol{g}_{m}^{(t)}$, the compressor $dpsign$ outputs $dpsign(\boldsymbol{g}_{m}^{(t)},\epsilon,\delta)$. The $i$-th entry of $dpsign(\boldsymbol{g}_{m}^{(t)},\epsilon,\delta)$ is given by
\begin{equation}\label{dpsignsgd}
\begin{split}
&dpsign(\boldsymbol{g}_{m}^{(t)},\epsilon,\delta)_{i} =
\begin{cases}
1, ~~~~~~~~~ \text{with probability $\Phi\big(\frac{(\boldsymbol{g}_{m}^{(t)})_{i}}{\sigma}\big)$} \\
-1,  ~~~~~~\text{with probability $1-\Phi\big(\frac{(\boldsymbol{g}_{m}^{(t)})_{i}}{\sigma}\big)$}\\
\end{cases}
\end{split}
\end{equation}
where $\Phi(\cdot)$ is the cumulative distribution function (cdf) of the normalized Gaussian distribution; $\sigma$ satisfies Equation~\ref{eq:Analytic_Gaussian}.
\end{Definition}

\begin{algorithm}
\caption{{\scriptsize DP-SIGN}SGD}
\label{algo:DP-SIGNSGD}
\begin{algorithmic}
\STATE \textbf{Input}: learning rate $\eta$, current global model vector $w^{(t)}$, each party's local data $D_m$, each party's gradient $\boldsymbol{g}_{m}^{(t)}$, {\scriptsize DP-SIGN}SGD compressor $dpsign(\cdot)$.
\STATE \textbf{Server:}
\STATE ~~\textbf{push} $\tilde{\boldsymbol{g}}^{(t)}= sign\big(\frac{1}{M}\sum_{m=1}^{M}dpsign(\boldsymbol{g}_{m}^{(t)})\big)$ \textbf{to} all parties
\STATE \textbf{Each party:}
\STATE ~~\textbf{compute} gradient: $\boldsymbol{g}_{m}^{(t)} \gets \texttt{SGD}(w^{(t)}, D_m)$ 
\STATE ~~\textbf{send} $dpsign(\boldsymbol{g}_{m}^{(t)})$ to server
\STATE ~~\textbf{update} $w^{(t+1)} = w^{(t)} - \eta\tilde{\boldsymbol{g}}^{(t)}$
\end{algorithmic}
\end{algorithm}

\begin{algorithm}[h]
\caption{Error-Feedback {\scriptsize DP-SIGN}SGD ({\scriptsize EF\text{-}DP-SIGNSGD})}
\label{algo:Error-feedback-DP-SIGNSGD}
\begin{algorithmic}
\STATE \textbf{Input}: learning rate $\eta$, current global model $w^{(t)}$, current residual error vector $\tilde{\boldsymbol{e}}^{(t)}$, each party's gradient $\boldsymbol{g}_{m}^{(t)}$, {\scriptsize DP-SIGN}SGD compressor $dpsign(\cdot)$, error decay rate $\lambda$.
\STATE \textbf{Server:}
\STATE ~~\textbf{push} $\tilde{\boldsymbol{g}}^{(t)} = sign\big(\frac{1}{M}\sum_{m=1}^{M}dpsign(\boldsymbol{g}_{m}^{(t)})+\tilde{\boldsymbol{e}}^{(t)}\big)$ \textbf{to} all parties
\STATE ~~\textbf{update residual error:}
\begin{equation*}\label{residualupdate}
\tilde{\boldsymbol{e}}^{(t+1)} = \lambda * \tilde{\boldsymbol{e}}^{(t)} + (1-\lambda)*[\frac{1}{M}\sum_{m=1}^{M}dpsign(\boldsymbol{g}_{m}^{(t)}) - \frac{1}{M}\tilde{\boldsymbol{g}}^{(t)}]
\end{equation*}
\STATE \textbf{Each party:}
\STATE ~~\textbf{compute} gradient: $\boldsymbol{g}_{m}^{(t)} \gets \texttt{SGD}(w^{(t)}, D_m)$ 
\STATE ~~\textbf{send} $dpsign(\boldsymbol{g}_{m}^{(t)})$ to server
\STATE ~~\textbf{update} $w^{(t+1)} = w^{(t)} - \eta\tilde{\boldsymbol{g}}^{(t)}$
\end{algorithmic}
\end{algorithm}

\begin{theorem}
The proposed compressor $dpsign(\cdot,\epsilon,\delta)$ is $(\epsilon,\delta)$-differentially private for any $\epsilon>0$, and $\delta \in (0,1)$.
\end{theorem}

\begin{proof}
We start from the one-dimension scenario and consider any $a,b$ that satisfy $||a - b||_{2} \leq \Delta_2$. Without loss of generality, assume that $dpsign(a,\epsilon,\delta)=dpsign(b,\epsilon,\delta)=1$. Then we have
\begin{equation}
\begin{split}
P(dpsign(a,\epsilon,\delta)=1) = \Phi\bigg(\frac{a}{\sigma}\bigg) = \int_{-\infty}^{a}\frac{1}{\sqrt{2\pi}\sigma}e^{-\frac{x^2}{2\sigma^2}}dx,\\
P(dpsign(b,\epsilon,\delta)=1) = \Phi\bigg(\frac{b}{\sigma}\bigg) = \int_{-\infty}^{b}\frac{1}{\sqrt{2\pi}\sigma}e^{-\frac{x^2}{2\sigma^2}}dx.
\end{split}
\end{equation}

\begin{equation}
\begin{split}
\frac{P(dpsign(a,\epsilon,\delta)=1)}{P(dpsign(b,\epsilon,\delta)=1)} = \frac{\int_{-\infty}^{a}e^{-\frac{x^2}{2\sigma^2}}dx}{\int_{-\infty}^{b}e^{-\frac{x^2}{2\sigma^2}}dx} = \frac{\int_{0}^{\infty}e^{-\frac{(x-a)^2}{2\sigma^2}}dx}{\int_{0}^{\infty}e^{-\frac{(x-b)^2}{2\sigma^2}}dx}.
\end{split}
\end{equation}
According to Theorem~\ref{Theorem:AGM}, given the parameters $\epsilon, \delta$ and $\sigma$, we can derive that $e^{-\epsilon} \leq \big|\frac{P(dpsign(a,\epsilon,\delta)=1)}{P(dpsign(b,\epsilon,\delta)=1)}\big| \leq e^{\epsilon}$ with probability at least $1-\delta$.

For the multi-dimension scenario, consider any vector $\boldsymbol{a}$ and $\boldsymbol{b}$ such that $||\boldsymbol{a} - \boldsymbol{b}||_{2} \leq \Delta_2$ and $\boldsymbol{v} \in \{-1,1\}^{d}$, we have
\begin{equation}
\begin{split}
\frac{P(dpsign(\boldsymbol{a},\epsilon,\delta)=\boldsymbol{v})}{P(dpsign(\boldsymbol{b},\epsilon,\delta)=\boldsymbol{v})} = \frac{\int_{D}e^{-\frac{||\boldsymbol{x}-\boldsymbol{a}||_{2}^2}{2\sigma^2}}d\boldsymbol{x}}{\int_{D}e^{-\frac{||\boldsymbol{x}-\boldsymbol{b}||_{2}^2}{2\sigma^2}}d\boldsymbol{x}},
\end{split}
\end{equation}
where $D$ is some integral area depending on $\boldsymbol{v}$. Similarly, it follows that $e^{-\epsilon} \leq \big|\frac{P(dpsign(\boldsymbol{a},\epsilon,\delta)=\boldsymbol{v})}{P(dpsign(\boldsymbol{b},\epsilon,\delta)=\boldsymbol{v})}\big| \leq e^{\epsilon}$ with probability at least $1-\delta$.
\end{proof}

\section{Performance Evaluation}
\label{sec:Performance}

\subsection{Dataset and Experimental Setup}
For experiments, we investigate two benchmark image datasets: MNIST\footnote{\url{http://yann.lecun.com/exdb/mnist/}} and CIFAR-10\footnote{\url{https://www.cs.toronto.edu/~kriz/cifar.html}}.
We run our experiments with 31 normal parties, and partition the training dataset according to the labels. The number of labels $c$ assigned to each party can be used as a metric to measure the data heterogeneity. We randomly sample examples from each label without replacement. Each party's data size depends on both $c$ and the size of the training data associated with each label. We compare our proposed algorithms with two baselines: {\scriptsize SIGN}SGD \cite{bernstein2018signsgd1} and FedAvg \cite{mcmahan2016communication}. For fair comparison, we set the same hyper-parameters (batch size as 256, local epoch as 1, and learning rate as 0.005 and 0.01 for MNIST and CIFAR-10 respectively) for all baselines. We set error decay rate as $\lambda=0.5$ when $c=\{2,4,10\}$, and omit it when $c=1$.

\begin{figure*}[tb]
\begin{subfigure}[ht]{0.23\textwidth}
\centering
\includegraphics[scale=0.26]{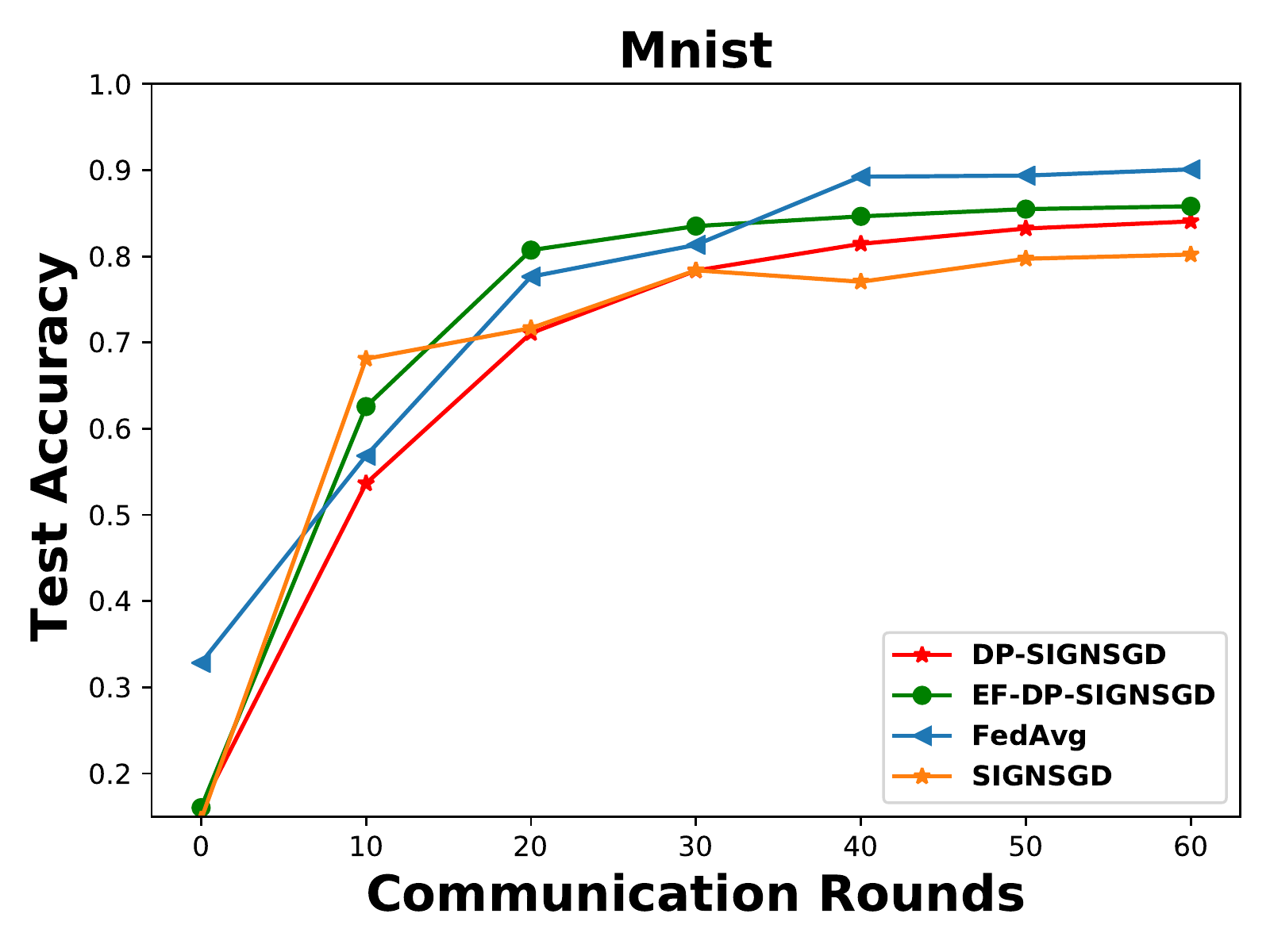}
 \end{subfigure}
\begin{subfigure}[ht]{0.23\textwidth}
\centering
\includegraphics[scale=0.26]{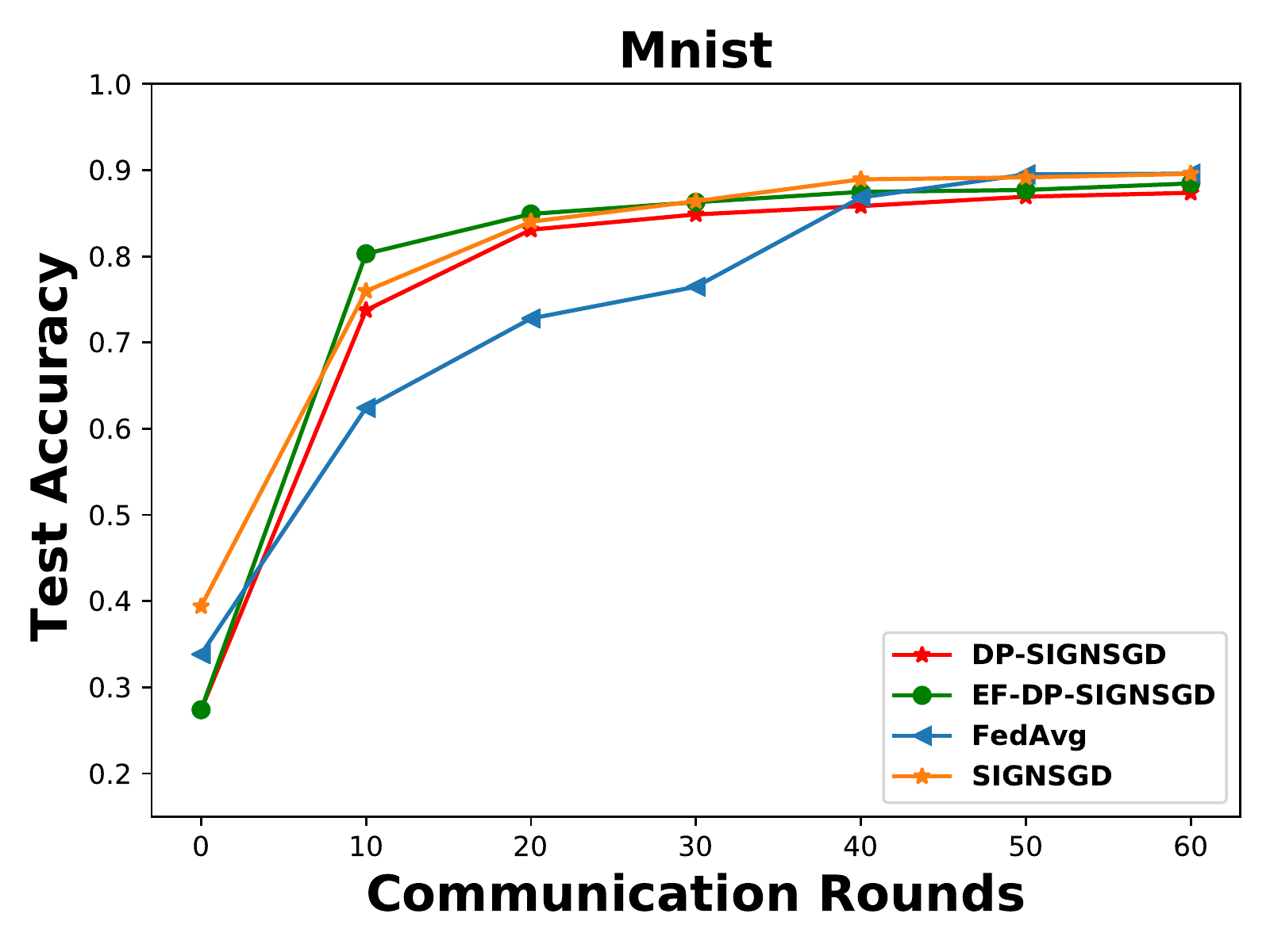}
 \end{subfigure}
 \begin{subfigure}[ht]{0.23\textwidth}
\centering
\includegraphics[scale=0.26]{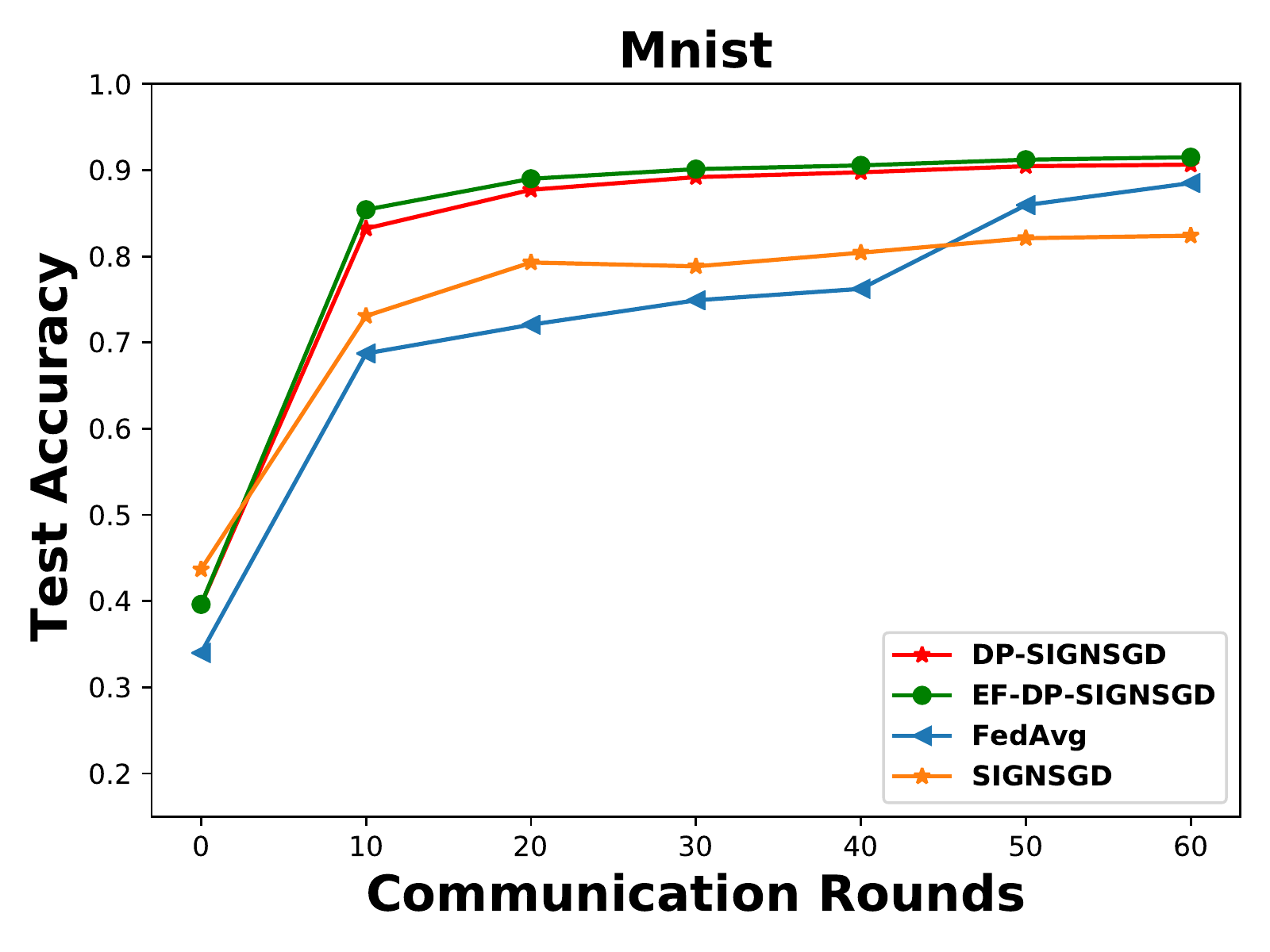}
 \end{subfigure}
 \begin{subfigure}[ht]{0.23\textwidth}
\centering
\includegraphics[scale=0.26]{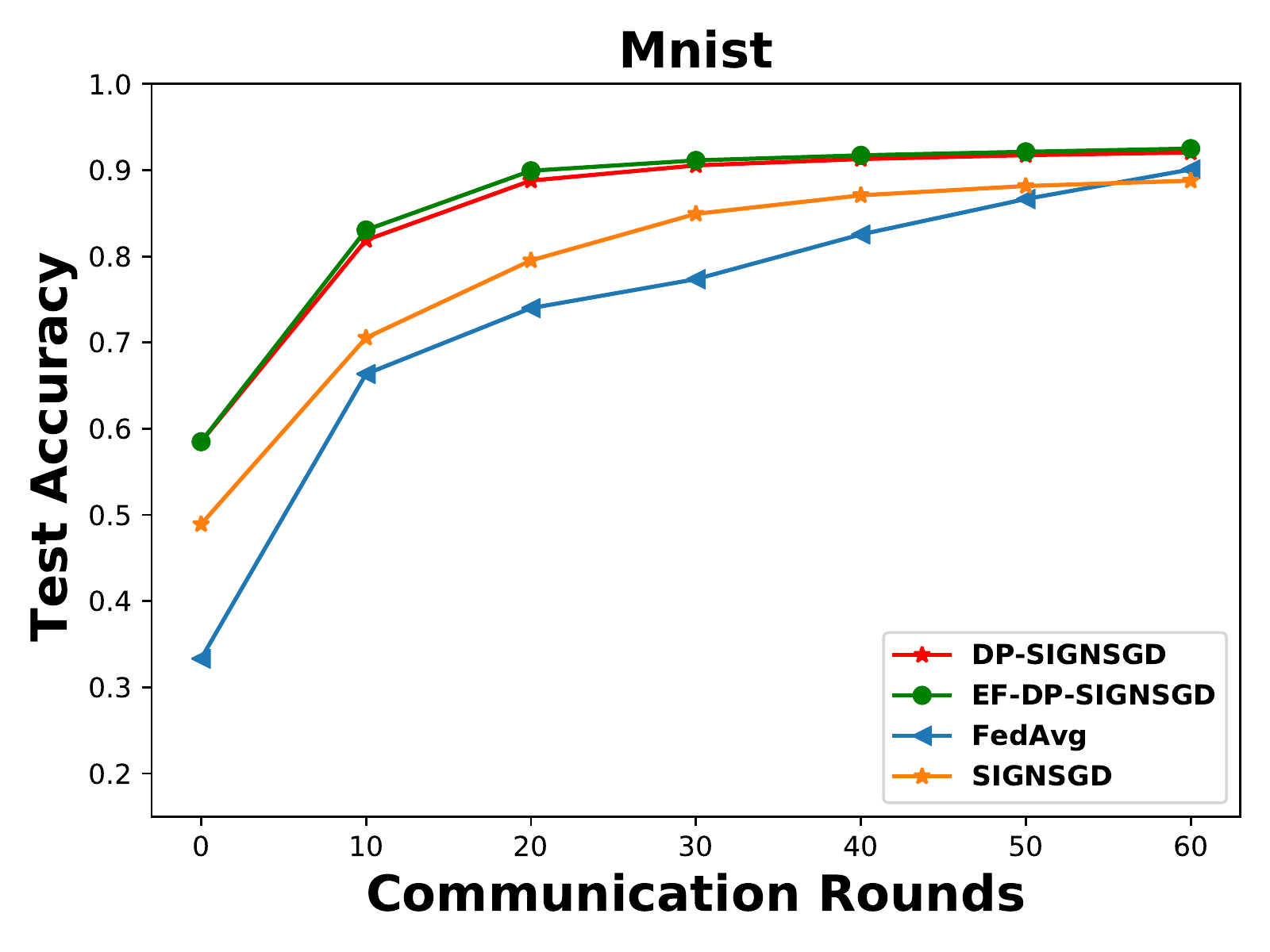}
 \end{subfigure}
\caption{MNIST test accuracy without Byzantine attackers. We fix $\epsilon=1, \delta = 10^{-5}$ for {\scriptsize DP-SIGN}SGD, and {\scriptsize EF-DP-SIGN}SGD. From left to right, each party has $c=\{1, 2, 4, 10\}$ classes.} 
\label{mnist_comparison}
\end{figure*}

In terms of local model architecture, we consider a simple feed-forward neural network with 64 hidden units. We remark that our main purpose is not to achieve state-of-the-art accuracy, but to validate the effectiveness of our proposed methods. For {\scriptsize DP-SIGN}SGD and {\scriptsize EF\text{-}DP-SIGN}SGD, we follow the DPSGD algorithm~\cite{abadi2016deep} to clip the gradient of each example within a fixed $L_{2}$ norm bound, such that the sensitivity $\Delta_{2}$ is bounded by a chosen threshold. We select $L_{2}$ norm bound as 4 and 1 for MNIST and CIFAR-10 respectively\footnote{code is available at: https://github.com/lingjuanlv/DP-SIGNSGD.git}.

\subsection{Experimental Results}

\textbf{Comparison with Baselines.}
As shown in Fig.~\ref{mnist_comparison}, for MNIST, {\scriptsize EF\text{-}DP-SIGN}SGD consistently outperforms {\scriptsize DP-SIGN}SGD in all cases. In addition, both {\scriptsize DP-SIGN}SGD and {\scriptsize EF\text{-}DP-SIGN}SGD converge well in spite of their biased nature. Similar results can also be observed from Fig.~\ref{cifar_comparison} for CIFAR-10.
We hypothesise that error-feedback mechanism in {\scriptsize EF\text{-}DP-SIGN}SGD takes effect and the injected noise in {\scriptsize EF\text{-}DP-SIGN}SGD and {\scriptsize DP-SIGN}SGD act as a regularization technique. Moreover, we notice that FedAvg does not necessarily outperform all the other baselines under the same hyper-parameter setting. FedAvg usually performs better by simultaneously increasing local training epochs and decreasing local mini-batch sizes, given a much larger starting learning rate with exponential decay~\cite{xu2020towards}. In terms of communication efficiency, compared with FedAvg using full-precision gradients, all the {\scriptsize SIGN}SGD based methods can reduce the communication overhead per round by 32$\times$, as all communication to and from the parameter server is compressed to one bit.

\begin{figure}
\centering
\includegraphics[scale=0.26]{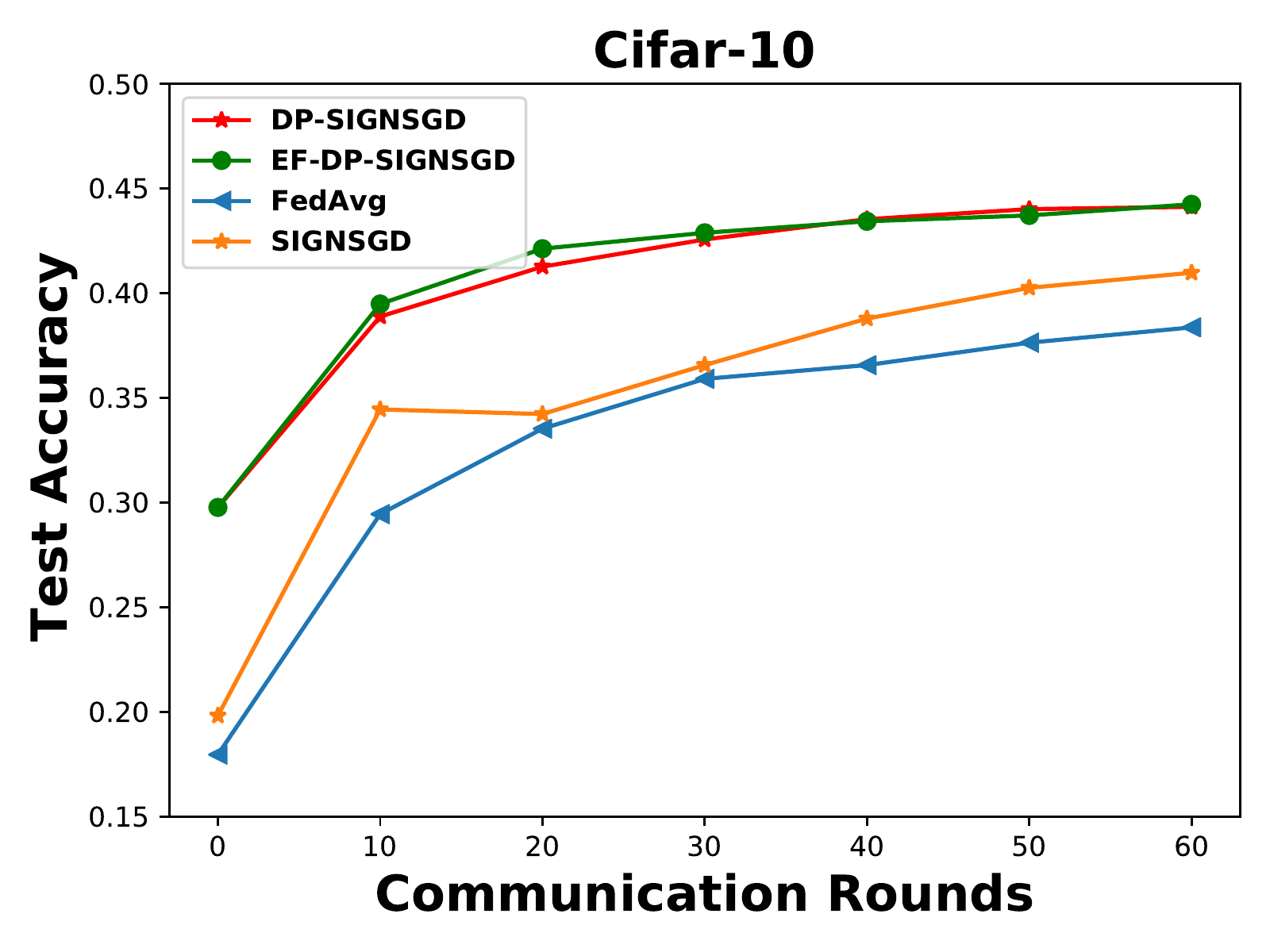}
\caption{CIFAR-10 test accuracy without Byzantine attackers. We fix $\epsilon=1, \delta = 10^{-5}$ for {\scriptsize DP-SIGN}SGD, and {\scriptsize EF-DP-SIGN}SGD. Each party has $c=10$ classes.}
\label{cifar_comparison}
\end{figure}

\textbf{Impact of Privacy Budget $\epsilon$}.
We examine the impact of privacy budget $\epsilon$ in Table~\ref{dp_sign_table} when $\delta = 10^{-5}$. Here $\epsilon$ measures per round privacy guarantee of each party. It can be observed that {\scriptsize EF-DP-SIGN}SGD outperforms {\scriptsize DP-SIGN}SGD for all $\epsilon$'s.

\begin{table}[ht]
\caption{The performance of MNIST {\scriptsize DP-SIGN}SGD and {\scriptsize EF-DP-SIGN}SGD under different $\epsilon$ when $\delta = 10^{-5}$, without Byzantine attackers. Each party has $c=4$ classes.}
\label{dp_sign_table}
\begin{center}
\begin{small}
\begin{tabular}{lcccccc}
 & $\epsilon$=0.05 & $\epsilon$=0.1 & $\epsilon$=0.5 & $\epsilon$=1 & $\epsilon$=2 \\
\midrule
{\scriptsize DP-SIGN}SGD & 75.54 & 82.14 & 89.33 & 90.64 & 91.95 \\ 
{\scriptsize EF-DP-SIGN}SGD & 78.68 & 84.93 & 90.32 &  91.50 & 92.84  \\ 
\bottomrule
\end{tabular}
\end{small}
\end{center}
\vskip -0.2in
\end{table}

\textbf{Byzantine Resilience}.
In addition to $M$ normal parties, we assume that there exist $B$ Byzantine attackers. Instead of following $dpsign$, the Byzantine attackers can take arbitrary compressors. In this work, we consider two types of adversaries: (1) \emph{Random} adversaries who randomise the sign of each coordinate of the gradients; (2) \emph{Negative} adversaries who invert their gradients. Note that we neglect the adversaries that arbitrarily rescale their gradients, as SIGN{SGD} does not aggregate the magnitudes of the values in the gradients but only the signs, which is inherently robust to all adversaries in this class. On the other hand, SGD is certainly not robust since an adversary could set the gradient to infinity and corrupt the entire model.

For \emph{Negative} adversaries, as the Byzantine attackers have access to the average gradients of all the $M$ normal parties (i.e., $\boldsymbol{g}_{j}^{(t)} = \frac{1}{M}\sum_{m=1}^{M}\boldsymbol{g}_{m}^{(t)}$, $\forall j \in [B]$), we assume that each Byzantine attacker $j$ shares the opposite signs of the true gradients, i.e., $byzantine\text{-}sign(\boldsymbol{g}_{j}^{(t)}) = -sign(\boldsymbol{g}_{j}^{(t)})$.

\begin{Remark}
Note that previous work has shown that SIGN{SGD} with Majority Vote~\cite{bernstein2018signsgd2} is more roust than MULTI-KRUM~\cite{blanchard2017machine}, so we only need to compare with SIGN{SGD} to show the efficacy of our methods. 
\end{Remark}

The test accuracy results of MNIST {\scriptsize EF-DP-SIGN}SGD with $\epsilon=1$ under varying Byzantine attackers are reported in Fig.~\ref{Byzantine_impact}. As we find that \emph{Negative} adversaries are always stronger than \emph{Random} adversaries, so we only show the Byzantine resilience results against \emph{Negative} adversaries for illustration purpose. As observed from Fig.~\ref{Byzantine_impact}, the Byzantine resilience of {\scriptsize EF-DP-SIGN}SGD keeps relatively stable even when the Byzantine attackers account for 40\% of all parties.

\begin{figure}
\centering
\includegraphics[scale=0.26]{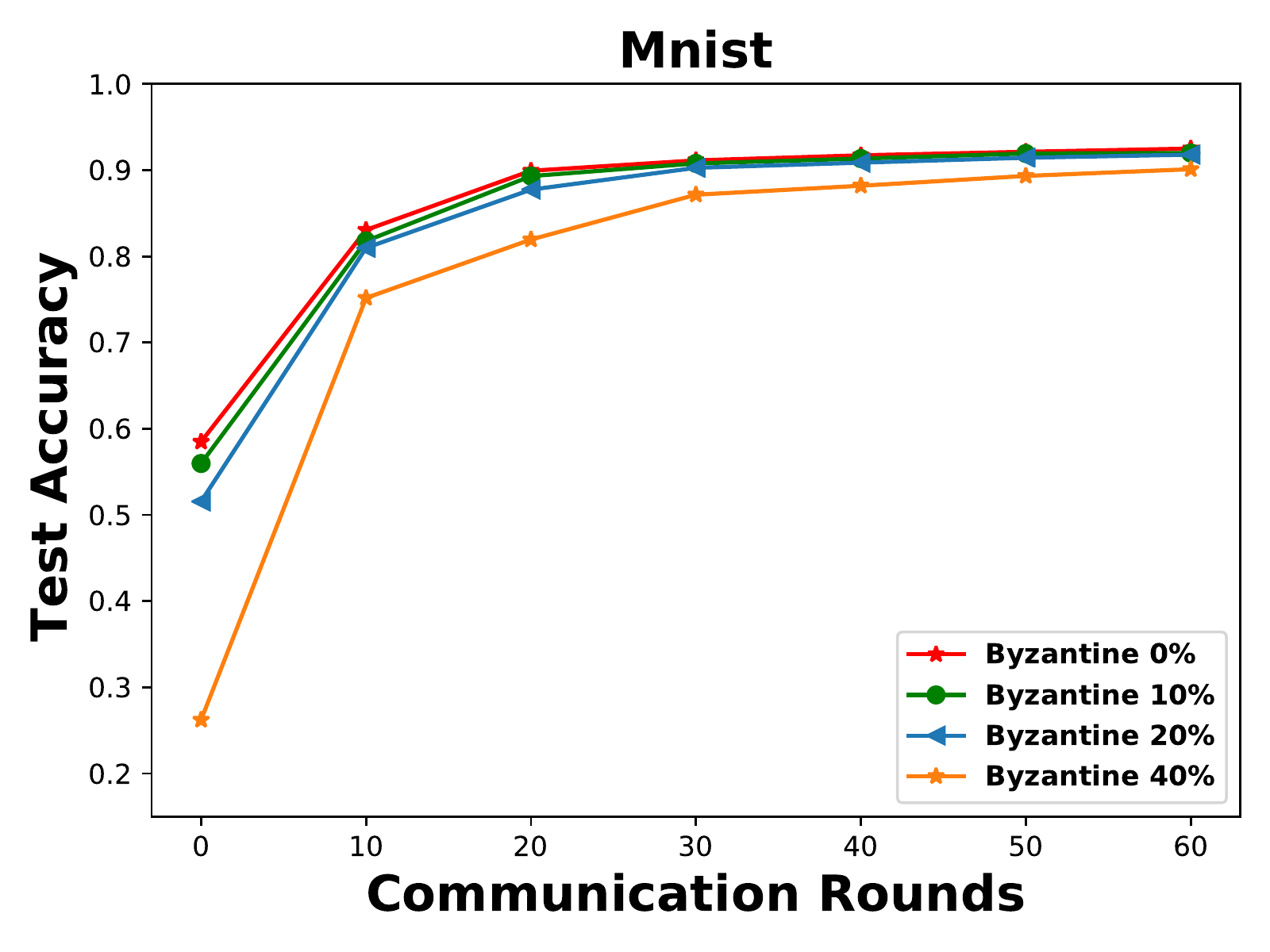}
\caption{The performance of {\scriptsize EF-DP-SIGN}SGD with different percentage of Byzantine parties. Each party has $c=10$ classes.}
\label{Byzantine_impact}
\end{figure}

\section{Conclusion}
\label{sec:Conclusion}
We propose an efficient, private, and Byzantine robust compressor by extending {\scriptsize SIGN}SGD to {\scriptsize DP-SIGN}SGD. We further incorporate the error-feedback mechanism to improve accuracy. We theoretically prove the privacy guarantee of the proposed algorithms, and empirically demonstrate that our proposed {\scriptsize DP-SIGN}SGD and {\scriptsize EF-DP-SIGN}SGD can achieve multiple goals simultaneously. We hope that our proposed algorithms can advance the FL algorithms towards the privacy-preserving, efficient, robust real-world applications. Theoretical convergence analysis would be the next-step work.

\bibliographystyle{IEEEbib}
\bibliography{icassp}
\end{sloppypar} 

\end{document}